\documentclass[research]{fcs}

\usepackage{bm}
\usepackage{times}
\usepackage{soul}
\usepackage{url}
\usepackage[utf8]{inputenc}
\usepackage{graphicx}
\usepackage{amsmath}
\usepackage{amsthm}
\usepackage{amssymb}
\usepackage{booktabs}
\usepackage{algorithmic}
\usepackage{multirow}
\usepackage{mwe}
\usepackage[ruled,linesnumbered,noend]{algorithm2e}

\usepackage{multirow}
\usepackage[normalem]{ulem}
\useunder{\uline}{\ul}{}
\usepackage{enumitem}
\usepackage{enumerate}
\usepackage{color}
\useunder{\uline}{\ul}{}
\usepackage{microtype}

\title{How Graph Convolutions Amplify Popularity Bias for Recommendation?}
\author[1]{Jiajia CHEN}
\author[1]{Jiancan WU}
\author[2]{Jiawei CHEN}
\author[3]{Xin XIN}
\author[4]{Yong LI}
\author*[1]{Xiangnan HE}
\address[1]{School of Information Science and Technology, University of Science and Technology of China, Hefei 230026, China}
\address[2]{School of Computer Science and Technology, Zhejiang University, Hangzhou 310058, China}
\address[3]{School of Computer Science and Technology, Shandong University, Qingdao 250100, China}
\address[4]{Department of Electronic Engineering, Tsinghua University, Beijing 100084, China}

\fcssetup{
  received       = {month dd, yyyy},
  accepted       = {month dd, yyyy},
  corr-email     = {xiangnanhe@gmail.com},
}
\begin{abstract}
Graph convolutional networks (GCNs) have  become prevalent in recommender system (RS) due to their superiority in modeling collaborative patterns. Although improving the overall accuracy, GCNs unfortunately amplify popularity bias --- tail items are less likely to be recommended. 
This effect prevents the GCN-based RS from making precise and fair recommendations, decreasing the effectiveness of recommender systems in the long run.

In this paper, we investigate how graph convolutions amplify the popularity bias in RS. Through theoretical analyses, we identify two fundamental factors: (1) with graph convolution (\textit{i.e.,} neighborhood aggregation), popular items exert larger influence than tail items on neighbor users, making the users move towards popular items in the representation space; (2) after multiple times of graph convolution, popular items would affect more high-order neighbors and become more influential. The two points make popular items get closer to almost users and thus being recommended more frequently.
To rectify this, we propose to estimate the amplified effect of popular nodes on each node's representation, and intervene the effect after each graph convolution. Specifically, we adopt clustering to discover highly-influential nodes and estimate the amplification effect of each node, then remove the effect from the node embeddings at each graph convolution layer. Our method is simple and generic --- it can be used in the inference stage to correct existing models rather than training a new model from scratch, and can be applied to various GCN models. We demonstrate our method on two representative GCN backbones LightGCN and UltraGCN, verifying its ability in improving the recommendations of tail items without sacrificing the performance of popular items. Codes are open-sourced \footnote{https://github.com/MEICRS/DAP}.
\end{abstract}
\begin{document}

\section{Introduction}
\label{intro}
Recommender system (RS) has become a key tool for personalization in today's Web, and also a research hotpot \cite{wu2022surveyaccuracy,chen2023bias}. Recently, methods based on GCN have become prevalent in RS since they can effectively encode collaborative filtering signal \cite{wang2019neural,wu2019neural,wu2019simplifying,he2020lightgcn,liu2021interest}. Though GCNs can improve the overall recommendation accuracy, we find a downside is they amplify the popularity bias --- popular items are more frequently recommended than tail items\footnote{We treat the most 20\% popular items in the dataset as popular items, and the rest are tail items.}, which means the tail items are not fairly treated by the algorithm, and the issue would be exacerbated in practical RSs due to  the user-system feedback loop. 

Figure \ref{tailratio} provides an evidence that GCNs intensify the popularity bias as the number of graph convolution layers increases. The backbone model is LightGCN \cite{he2020lightgcn} using BPR loss \cite{SteffenRendle2009BPRBP}, a representative and competitive GCN-based recommender. With the increasing of graph convolution layers, the overall accuracy metric (\textit{i.e.,} Recall@20) gradually increases, but the ratio of recommending tail items (i.e., TR@20) drops significantly. This means the improvements of performing graph convolutions is at the expense of tail items, pointing to the amplified popularity bias issue of GCNs~\cite{wei2021model}. Although there has been recent work in alleviating the popularity bias in RS~\cite{zheng2021disentangling, wei2021model, zhang2021causal}, they do not reveal how GCN amplifies the bias (\textit{i.e.,} the mechanism) and their solutions are generally applied to RS models rather than tailored for GCNs. 

\begin{figure}  
\centering
    \includegraphics[scale=0.4]{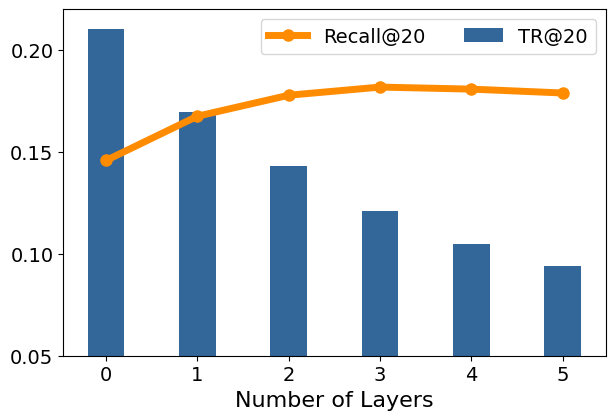}
    \caption{Performance change of LightGCN with different graph convolution layers on Gowalla. Recall@20 and TR@20 stand for the overall recall score and the ratio of tail items in the top-20 recommendation list, respectively.}
    \label{tailratio}
\vspace{-10pt}
\end{figure}

In this work, we aim to answer the unsolved questions "\textit{How do graph convolutions amplify popularity bias in RS?}" and "\textit{How to address the issue with minimum changes on the training progress?}" To reveal the mechanism of GCNs amplifying the bias, we conduct theoretical studies by analyzing how popular items influence other nodes. The main findings are twofold: (1) with graph convolution, popular items have a greater impact than tail items on neighbor users, making the users position closer to popular items in the representation space; (2) after multiple times of graph convolution, popular items become more influential by affecting more high-order neighbor users. The theoretical analyses confirm the inherent defect of GCNs in over-recommending popular items.

To alleviate the popularity bias in RS, two lines of research work have been conducted. The first is intervening the training process to eliminate the bias effect. For example,  propensity-based methods~\cite{liang2016causal} impose lower weights on popular items in training loss and  causal methods~\cite{zheng2021disentangling, wei2021model} model the causal effect of the bias on the model training. Those methods need to retrain the backbone model, making the solutions costly to use in practical RSs especially to correct the already-deployed GCN models. In contrast, another line is revising the inference stage in a post-hoc way. For example, \cite{abdollahpouri2019managing} performs personalized re-ranking on the generated candidates of the RS model to suppress popular items. However, both lines work on the general issue of popularity bias in recommendation, leaving how GCN model suffers from and amplifies the popularity bias untouched.  

Towards the research gap, we propose a new problem to solve --- rectifying the popularity bias of GCN model in the inference stage. Such a solution can be used to revise already-trained GCNs, thus is easier to deploy in practice than training a new model. Given a GCN model, we first cluster the node representations to automatically discover the highly-influential nodes. Then, the amplification of popularity bias for each node within its cluster is estimated based on the prior theoretical analyses. Thereafter, the amplification effect in the node representation can be intervened to control the bias. This post-hoc method can be easily deployed in practice to correct existing GCN models and promote the recommendations of tail items without sacrificing the performance of popular items. To summarize, this work makes the following  contributions: 

\begin{itemize}
\item Providing in-depth theoretical analyses to interpret the popularity bias amplification problem in GCN-based recommenders;  
\item Developing a new method working at each graph convolution layer in the inference stage to correct the popularity bias for GCN;
\item Conducting extensive experiments on three real datasets to demonstrate the effectiveness of our method on LightGCN and UltraGCN backbones. 
\end{itemize}
\section{Preliminaries}

\begin{table}[t]
\caption{Main notations used in the paper.}
\label{table:0}
\begin{tabular}{l|l}
\hline
$\mathcal{U}$, $\mathcal{I}$  & User set, item set \\ \hline
$\mathcal{N}_u$, $\mathcal{N}_i$  & The one-order neighbors of user $u$ or item $i$ \\ \hline
$d_u$, $d_i$  & The degree of user $u$ or item $i$ \\ \hline
$\mathbf{e}_u^{(l)}$, $\mathbf{e}_i^{(l)}$  & \parbox[c]{6.5cm}{The embedding of user $u$ or item $i$ at the $l$-th graph convolution layer}  \\ \hline
$\mathcal{L}^{ui}$  & The individual loss term of an interaction $(u,i)$ \\ \hline
$\mathcal{C}^{(l)}_p$  & \parbox[c]{6.5cm}{A set of nodes in the $p$-th cluster obtained by using Kmeans given the embeddings $\mathbf{E}^{(l)}$}  \\ \hline
$\mathbb{H}_v^{(l)}$  & A set of nodes $\{j\in\mathcal{C}^{(l)}_p, d_j > d_v|v\in\mathcal{C}^{(l)}_p\}$ \\ \hline
$\mathbb{L}_v^{(l)}$ & A set of nodes $\{j\in\mathcal{C}^{(l)}_p, d_j < d_v|v\in\mathcal{C}^{(l)}_p\}$ \\ \hline
$\hat{\pmb{\vartheta}}_{H_v}^{(l)}$ & \parbox[c]{6.5cm}{The pooling representations after normalization of $\mathbb{H}_v^{(l)}$} \\ \hline
$\hat{\pmb{\vartheta}}_{L_v}^{(l)}$ & \parbox[c]{6.5cm}{The pooling representations after normalization of $\mathbb{L}_v^{(l)}$} \\ \hline
\end{tabular}
\end{table}

Suppose that there are a set of users and a set of items $\mathcal{U}=\{u_1, u_2, \cdots, u_M\}$, $\mathcal{I}=\{i_1, i_2, \cdots, i_N\}$ in a dataset $D$. Let $y_{ui}=1$ be the positive label if the user $u$ has interacted with the item $i$, otherwise $y_{ui}=0$.
We can construct a user-item bipartite graph $\mathcal{B} = (\mathcal{V}, \mathcal{E})$ based on the interaction history, where $\mathcal{V}$ consists of the set of user and item nodes, and $\mathcal{E}$ denotes the set of edges. If $y_{ui}=1$, there is an edge between the user $u$ and the item $i$. 

Recently, many studies opt for powerful GCNs to learn user and item node representations \cite{KelongMao2021UltraGCNUS, kipf2016semi}. Particularly, we introduce LightGCN \cite{he2020lightgcn}, which is neat and well represents the GCN-based recommenders. 
One graph convolution block of LightGCN can be expressed as:
\begin{equation}
\label{2110701}
    \textbf{e}_u^{(l)} = \sum\limits_{i \in \mathcal{N}_u} \frac{1}{\sqrt{d_u} \sqrt{d_i}}\textbf{e}_i^{(l-1)}, ~~~~
    \textbf{e}_i^{(l)} = \sum\limits_{u \in \mathcal{N}_i} \frac{1}{\sqrt{d_i} \sqrt{d_u}}\textbf{e}_u^{(l-1)},
\end{equation}
where $d_u$ ($d_i$) is the degree of user $u$ (item $i$) in the graph $\mathcal{B}$, $\mathcal{N}_u$ ($\mathcal{N}_i$) is the one-order neighbor nodes of the user $u$ (item $i$), $\textbf{e}^{(0)}$ is an ID embedding of a user or an item. 
After stacking several graph convolution layers, LightGCN combines the embeddings obtained at each layer to form the final representation $\mathbf{e}$ of a node.
Thereafter, the model prediction is defined as the inner product of user and item final representations, \textit{i.e.}, $\hat{y}_{ui} = \textbf{e}_u^{\top} \textbf{e}_i$. Another representative work is UltraGCN~\cite{KelongMao2021UltraGCNUS} which skips infinite layers of graph convolution. We also conduct experiments on this model to verify the generality of our method.

To optimize model parameters, prior work usually frames it as a supervised learning task, and utilize a pointwise or a pairwise loss for model training, which can be summarized by the following formula:
$\mathcal{L}=\sum_{(u,i)\in D} \mathcal{L}^{ui}$,
where $\mathcal{L}^{ui}$ is the individual loss term of an interaction $(u,i)$. Without loss of generality,
we investigate the popularity bias amplification problem in GCNs based on BCE loss \cite{ChristopherJohnson2022LogisticMF} in the next section. The formal formulation of BCE loss is
\begin{equation}
\label{20220703bceloss}
\mathcal{L}^{ui}  = -[y_{ui}\ln\sigma(\hat{y}_{ui}) + (1-y_{ui})\ln(1-\sigma(\hat{y}_{ui}))].
\end{equation}

\section{METHODOLOGY}
In this part, we attempt to analyze and resolve the amplified popularity bias of GCNs.

\subsection{Popularity Bias Amplification in GCNs}
To understand why GCNs amplify the popularity bias, we conduct theoretical analyses and empirical experiments on GCNs: (1) we start with defining the influence between nodes based on the training loss; (2) we prove that popular items with higher degrees exert larger influence on neighbor users than tail items with lower-degrees; (3) we reveal that popular items commonly have higher probabilities of being recommended by users after representation updating and graph convolution in GCNs.

Concisely, we take the $\textbf{e} = \textbf{e}^{(L)}$ at $L$-th graph convolution layer in Eq. \eqref{2110701} as the final representation of each node.
Next, we give the definition of the influence of a user-item pair loss on their neighbors exploiting the concept of influence functions \cite{koh2017understanding, xu2018representation, tang2020investigating}:

\textbf{Definition 1} (Influence of an observed interaction on a node's representation learning): Suppose that $(u,i)$ is an observed interaction, \textit{i.e.}, there is an edge between node $u$ and node $i$, some other node $k$ is reachable from node $u$ or $i$, then the influence of an interaction $(u,i)$ on node $k$ is defined as $\frac{\partial \mathcal{L}^{ui}}{\partial \mathbf{e}_k}$.

Without loss of generality, we mainly consider $y_{ui}=1$ in BCE loss and have
\begin{equation}
\label{20220806eq1}
\begin{aligned}
\frac{\partial \mathcal{L}^{ui}}{\partial \mathbf{e}_k} &= -\frac{\partial \ln\sigma(\hat{y}_{ui})}{\partial \mathbf{e}_k}
 = -\frac{\ln\sigma(\hat{y}_{ui})}{\partial \sigma(\hat{y}_{ui})} \cdot \frac{\partial \sigma(\hat{y}_{ui})}{\partial \hat{y}_{ui}} \cdot \frac{\partial \hat{y}_{ui}}{\partial \mathbf{e}_k}\\
& = -\frac{1}{\sigma(\hat{y}_{ui})}\cdot\sigma(\hat{y}_{ui})(1-\sigma(\hat{y}_{ui}))\cdot\frac{\partial \hat{y}_{ui}}{\partial \mathbf{e}_k}\\
& = -[1 - \sigma(\hat{y}_{ui})]\frac{\partial \hat{y}_{ui}}{\partial \mathbf{e}_k}\\
& = - \lambda_{ui}\frac{\partial \hat{y}_{ui}}{\partial \mathbf{e}_k},
\end{aligned}
\end{equation}
where $\hat{y}_{ui}$ is the prediction between the user $u$ and the item $i$ and $0<\lambda_{ui} <1$.

\textbf{Definition 2} (Influence of a node on another node's representation learning): Suppose that a node $i$ can reach a neighbor node $k$ on a graph\footnote{Without loss of generality, we ignore whether a node is a user or an item.}. The influence of the loss for the target node $i$ on the node $k$ is defined as $\frac{\partial \mathcal{L}_i}{\partial \mathbf{e}_k}$, where $\mathcal{L}_i = \sum_{j\in \mathcal{N}_i} \mathcal{L}^{ij}$. 

In fact, the influence provides a fine-grained scrutiny of updating information of each node at the lens of gradient derivation. We then have the following lemma.
\begin{lemma}
\label{21107thm2}
If node $i$ with degree $d_i$ can reach node $k$ after stacking $L$ layers of graph convolution, then the influence of node $i$ on node $k$ follows 
\begin{equation}
\label{influence_i}
    \mathbb{E}\left (\frac{\partial \mathcal{L}_i}{\partial \textbf{e}_k}\right) \propto
  -d_i^{\frac{3}{2}} \pmb{\vartheta}_i,
\end{equation}
where $\mathbb{E}(\cdot)$ is the expectation, $\pmb{\vartheta}_i = \mathbb{E}\left(\sum\limits_{p=1}^{\Phi_j} \prod\limits_{l=L-1}^1 \frac{1}{\sqrt{d_{p^l}}}\textbf{e}_j\right)$ represents the expectation on paths starting from a neighbor node $j$ of the node $i$ to the node $k$, $\Phi_j$ is the set of all $(L-1)$-length paths from the node $j$ to the node $k$, $p^{L-1}$ and $p^1$ are the node $j$ and the node $k$, respectively. 
\end{lemma}
\begin{proof}
According to Eq. \eqref{20220806eq1} and Eq. \eqref{2110701}, we obtain
\begin{equation}
\begin{aligned}
\label{220620eq1}
    \frac{\partial \mathcal{L}_i}{\partial \mathbf{e}_k} &= \sum\limits_{j\in \mathcal{N}_i}\frac{\partial \mathcal{L}^{ij}}{\partial \mathbf{e}_k}
    =-\sum\limits_{j\in \mathcal{N}_i}\lambda_{ij}\frac{\partial \hat{y}_{ij}}{\partial \mathbf{e}_k} = -\sum\limits_{j\in \mathcal{N}_i}\lambda_{ij}\frac{\partial\mathbf{e}_i}{\partial\mathbf{e}_k}\mathbf{e}_j\\
    &= -\sum\limits_{j\in \mathcal{N}_i}\lambda_{ij} \sum\limits_{p=1}^{\Phi}\prod\limits_{l=L}^{1} \frac{1}{\sqrt{d_{p^l}}}\mathbf{e}_j,
\end{aligned}
\end{equation}
where $p^{L}$ is the node $i$, $\Phi$ is the set of all $L$-length random paths on the graph from nodes $i$ to $k$, and we omit the transpose symbol of the partial derivative for brevity. 
Further, we have

\begin{equation}
\begin{aligned}
\mathbb{E}\left (\frac{\partial \mathcal{L}_i}{\partial \textbf{e}_k}\right) &= -\mathbb{E}\left( \sum\limits_{j\in \mathcal{N}_i}\lambda_{ij} \sum\limits_{p=1}^{\Phi}\prod\limits_{l=L}^{1} \frac{1}{\sqrt{d_{p^l}}}\mathbf{e}_j\right)
\\
&\propto -\sqrt{d_i} \mathbb{E}\left(\sum\limits_{j\in \mathcal{N}_i} \sum\limits_{p=1}^{\Phi_j}  \prod\limits_{l=L-1}^1 \frac{1}{\sqrt{d_{p^l}}}\textbf{e}_j\right)
\\
&\approx -d_i^{\frac{3}{2}}\mathbb{E}\left(\sum\limits_{p=1}^{\Phi_j}  \prod\limits_{l=L-1}^1 \frac{1}{\sqrt{d_{p^l}}}\textbf{e}_j\right)
 = -d_i^{\frac{3}{2}} \pmb{\vartheta}_i.
\end{aligned}
\label{20220126eq1}
\end{equation}
Finish the proof.
\end{proof}

We visualize the results of $d^{\frac{3}{2}}\Vert\pmb{\vartheta}\Vert$ and $\Vert\pmb{\vartheta}\Vert$ in log scale in Figure \ref{theta_vis}. Items are evenly divided into several groups in ascending order of their degrees. For each item group, we show its average $\ln(d^{\frac{3}{2}}\Vert\pmb{\vartheta}\Vert)$ and $\ln(\Vert\pmb{\vartheta}\Vert)$ in terms of the two-hop neighbor nodes (\textit{i.e.,} $L=2$ in Lemma \ref{21107thm2}) when training LightGCN. As we see, $\ln(d^{\frac{3}{2}}\Vert\pmb{\vartheta}\Vert)$ enlarges as degree increases. Compared to $\ln(d^{\frac{3}{2}}\Vert\pmb{\vartheta}\Vert)$, $\ln(\Vert\pmb{\vartheta}\Vert)$ is relatively smaller and flat across at different degrees. This illustrates that the degree of nodes plays a crucial role in the influence. Following the assumption in \cite{tang2020investigating}, we posit $\Vert\pmb{\vartheta}_i\Vert = \phi$ for any node $i$. 

\noindent \textbf{Conclusion 1:} If nodes $r$ and $s$ with $d_r > d_s$ both can reach node $t$ after $L$-hop neighborhood aggregation in a graph, we could have
\begin{equation}
\left\|\mathbb{E}\left(\frac{\partial \mathcal{L}_r}{\partial \textbf{e}_t}\right)\right\| > \left\|\mathbb{E}\left(\frac{\partial \mathcal{L}_s}{\partial \textbf{e}_t}\right) \right\|.
\end{equation}
It suggests that the nodes with higher degrees exert larger influence on L-hop neighbor nodes than lower-degree nodes in the training stage of GCN-based models. In other words, the popular items dominate the updating information of neighbor users. 
As a result, popular items would make reachable users get closer to them in the representation space.
Based on the results in the upcoming lemma, we further prove that the popular items tend to have higher probabilities of being recommended by users.

\begin{figure}
\includegraphics[scale=0.42]{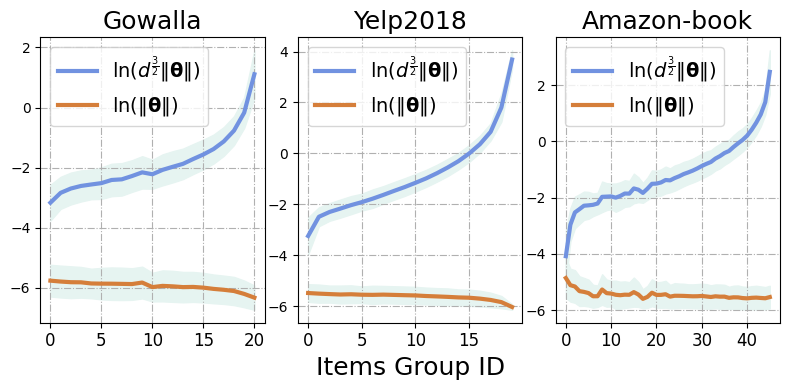}
\caption{Average $\Vert\pmb{\theta}\Vert$ and $d^{\frac{3}{2}}\Vert\pmb{\theta}\Vert$ in each items group. Items are sorted into groups in ascending order of their degrees.}  
\label{theta_vis}
\end{figure}

\begin{figure}
\includegraphics[scale=0.42]{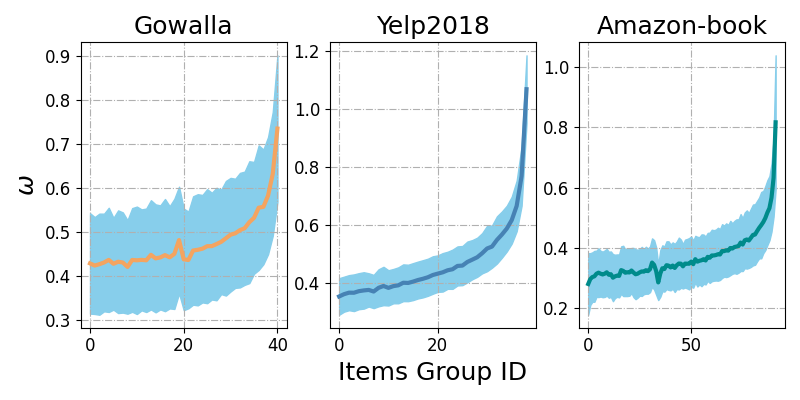}
\caption{The average aggregation weight of one-order neighbor users in each items group. Items are sorted into groups in ascending order of their degrees.}  
\label{factor_omega}
\vspace{-0.4cm}
\end{figure}

\begin{lemma}
\label{220703thm3}
 Suppose that the items $r$ and $s$ could reach the user $t$ by stacking $L$ layers of graph convolution, and $d_r>d_s$. After $L-1$ rounds of graph convolution, the expectation of prediction difference between the two items with regard to the user $t$ is $\mathbb{E}\left[\mathbf{e}_t^{\top}(\mathbf{e}_r-\mathbf{e}_s)\right]$. By performing the $L$-th graph convolution and after the representation of user $t$ updates the gradients $\left(\frac{\partial \mathcal{L}_r}{\partial \mathbf{e}_t} + \frac{\partial L_s}{\partial \mathbf{e}_t}\right)$, the prediction difference becomes larger, \text{i.e.,}
\begin{equation}
\label{20220720eq2}
    \mathbb{E}\left\{\left[\mathbf{e}_t - \left(\frac{\partial \mathcal{L}_r}{\partial \mathbf{e}_t} + \frac{\partial \mathcal{L}_s}{\partial \mathbf{e}_t}\right)\right]^{\top}(\mathbf{e}_r - \mathbf{e}_s)\right\} \geq \mathbb{E}\left[\mathbf{e}_t^{\top}(\mathbf{e}_r-\mathbf{e}_s)\right].
\end{equation}
Further, after $e_r$ and $e_s$ updates, the prediction difference continues to enlarge,
\begin{equation}
\begin{aligned}
    &\mathbb{E}\left\{\left[\mathbf{e}_t - \left(\frac{\partial \mathcal{L}_r}{\partial \mathbf{e}_t} + \frac{\partial \mathcal{L}_s}{\partial \mathbf{e}_t}\right)\right]^{\top}(\mathbf{e}'_r - \mathbf{e}'_s)\right\}
    \\ 
    &\geq \mathbb{E}\left\{\left[\mathbf{e}_t - \left(\frac{\partial \mathcal{L}_r}{\partial \mathbf{e}_t} + \frac{\partial \mathcal{L}_s}{\partial \mathbf{e}_t}\right)\right]^{\top}(\mathbf{e}_r - \mathbf{e}_s)\right\}.
\end{aligned}
\label{20220724eq1}
\end{equation}
where $\mathbf{e}'_r$ and $\mathbf{e}'_s$ are the representations of the items $r$ and $s$ after aggregating the user $t$, respectively.
\end{lemma}
\begin{proof}
After the influence of the items $r$ and $s$ propagate to the user $t$ by stacking $L$ graph convolution layers, this changes the representation of user $t$ from $\mathbf{e}_t$ to $\mathbf{e}_t - \left(\frac{\partial \mathcal{L}_r}{\partial \mathbf{e}_t} + \frac{\partial \mathcal{L}_s}{\partial \mathbf{e}_t}\right)$. Assume that the influence $\mathbb{E}\left(\frac{\partial\mathcal{L}_r}{\partial\mathbf{e}_t}\right) = -\upsilon_1 d_r^{\frac{3}{2}} \pmb{\vartheta}_r$ ($\upsilon_1>0$) and $\frac{\pmb{\vartheta}_r}{\Vert\pmb{\vartheta}_r\Vert} = \upsilon_2 \frac{\mathbf{e}_r}{\Vert\mathbf{e}_r\Vert}$ ($\upsilon_2 > 0$ for the local 
homogeneity), likewise, $\mathbb{E}\left(\frac{\partial\mathcal{L}_s}{\partial\mathbf{e}_t}\right) = -\upsilon_1 d_s^{\frac{3}{2}} \pmb{\vartheta}_s$ and $\frac{\pmb{\vartheta}_s}{\Vert\pmb{\vartheta}_s\Vert} = \upsilon_2 \frac{\mathbf{e}_s}{\Vert\mathbf{e}_s\Vert}$. Now we calculate the prediction difference between the items $r$ and $s$ on the user $t$,
\begin{equation}
\begin{aligned}
&\mathbb{E}\left\{\left[\mathbf{e}_t - \left(\frac{\partial \mathcal{L}_r}{\partial \mathbf{e}_t} + \frac{\partial \mathcal{L}_s}{\partial \mathbf{e}_t}\right)\right]^{\top}(\mathbf{e}_r - \mathbf{e}_s)\right\} \\
=&\mathbb{E}[\mathbf{e}_t^{\top}(\mathbf{e}_r - \mathbf{e}_s)] \\&+ \upsilon_1\mathbb{E}\left[d_r^{\frac{3}{2}}\pmb{\vartheta}_r^{\top} \mathbf{e}_r - d_r^{\frac{3}{2}}\pmb{\vartheta}_r^{\top}\mathbf{e}_s +d_s^{\frac{3}{2}}\pmb{\vartheta}_s^{\top}\mathbf{e}_r - d_s^{\frac{3}{2}}\pmb{\vartheta}_s^{\top}\mathbf{e}_s\right]\\
=& \mathbb{E}[\mathbf{e}_t^{\top}(\mathbf{e}_r - \mathbf{e}_s)]+\upsilon_1\upsilon_2 \mathbb{E}\left[d_r^{\frac{3}{2}}\frac{\Vert\pmb{\vartheta}_r\Vert}{\Vert\mathbf{e}_r\Vert} \mathbf{e}_r^{\top}\mathbf{e}_r -d_r^{\frac{3}{2}}\frac{\Vert\pmb{\vartheta}_r\Vert}{\Vert\mathbf{e}_r\Vert} \mathbf{e}_r^{\top}\mathbf{e}_s\right.\\ &\left.+d_s^{\frac{3}{2}}\frac{\Vert\pmb{\vartheta}_s\Vert}{\Vert\mathbf{e}_s\Vert} \mathbf{e}_s^{\top}\mathbf{e}_r - d_s^{\frac{3}{2}}\frac{\Vert\pmb{\vartheta}_s\Vert}{\Vert\mathbf{e}_s\Vert} \mathbf{e}_s^{\top}\mathbf{e}_s\right]\\
=& \mathbb{E}[\mathbf{e}_z^{\top}(\mathbf{e}_r - \mathbf{e}_s)] \\&+
\upsilon_1\upsilon_2\phi\mathbb{E}\left[d_r^{\frac{3}{2}}(\Vert\mathbf{e}_r\Vert-\frac{\rho}{\Vert\mathbf{e}_r\Vert}) - d_s^{\frac{3}{2}}(\Vert\mathbf{e}_s\Vert-\frac{\rho}{\Vert\mathbf{e}_s\Vert})\right],
\end{aligned}
\end{equation}
where $\rho = \mathbf{e}_r^{\top} \mathbf{e}_s$ and $\rho\leq \Vert\mathbf{e}_r\Vert\Vert\mathbf{e}_s\Vert$. Since the magnitude of the node representation increases as its degree increases \cite{wu2022effectiveness}, for $d_r > d_s$, there is $\Vert\mathbf{e}_r\Vert > \Vert\mathbf{e}_s\Vert$ generally. Let $\Vert\mathbf{e}_r\Vert = \kappa \Vert\mathbf{e}_s\Vert$ with $\kappa>1$. 
Therefore, we have 
\begin{equation}
\label{20220720eq1}
\begin{aligned}
    &d_r^{\frac{3}{2}}(\Vert\mathbf{e}_r\Vert-\frac{\rho}{\Vert\mathbf{e}_r\Vert}) - d_s^{\frac{3}{2}}(\Vert\mathbf{e}_s\Vert-\frac{\rho}{\Vert\mathbf{e}_s\Vert})\\
    =&(d_r^{\frac{3}{2}}\kappa-d_s^{\frac{3}{2}})\Vert\mathbf{e}_s\Vert - (d_r^{\frac{3}{2}}-d_s^{\frac{3}{2}}\kappa)\frac{\rho}{\Vert\mathbf{e}_r\Vert},
\end{aligned}
\end{equation}
since $d_r^{\frac{3}{2}}\kappa-d_s^{\frac{3}{2}} > d_r^{\frac{3}{2}}-d_s^{\frac{3}{2}}\kappa$ and $\Vert\mathbf{e}_s\Vert \geq \frac{\rho}{\Vert\mathbf{e}_r\Vert}$, thus Eq. \eqref{20220720eq1} $>$ 0.
Based on this, we derive the expression of \eqref{20220720eq2}. 

Furthermore, after the items $r$ and $s$ aggregate the information of the user $t$, we obtain
\begin{equation}
\begin{aligned}
\mathbf{e}'_r &= \mathbf{e}_r + \omega_{rt}\left[\mathbf{e}_t - \left(\frac{\partial \mathcal{L}_r}{\partial \mathbf{e}_t} + \frac{\partial \mathcal{L}_s}{\partial \mathbf{e}_t}\right)\right]
=\mathbf{e}_r + \omega_{rt}\Tilde{\mathbf{e}}_t,
\end{aligned}
\end{equation}
where 
$\omega_{rt}$ is the weight of aggregation.
Likewise,
\begin{equation}
\begin{aligned}
\mathbf{e}'_s &= \mathbf{e}_s + \omega_{st}\left[\mathbf{e}_t - \left(\frac{\partial \mathcal{L}_r}{\partial \mathbf{e}_t} + \frac{\partial \mathcal{L}_s}{\partial \mathbf{e}_t}\right)\right]
= \mathbf{e}_s + \omega_{st}\Tilde{\mathbf{e}}_t.
\end{aligned}
\end{equation}
Now we calculate the rating difference again,
\begin{equation}
\begin{aligned}
&\mathbb{E}\left\{\left[\mathbf{e}_t - \left(\frac{\partial \mathcal{L}_r}{\partial \mathbf{e}_t} + \frac{\partial \mathcal{L}_s}{\partial \mathbf{e}_t}\right)\right]^{\top}(\mathbf{e}'_r - \mathbf{e}'_s)\right\} \\
=& \mathbb{E}\left\{\Tilde{\mathbf{e}}^{\top}_t[(\mathbf{e}_t - \mathbf{e}_s) + (\omega_{rt}-\omega_{st})\Tilde{\mathbf{e}}_t] \right\}\\
=& \mathbb{E}\left[\Tilde{\mathbf{e}}^{\top}_t(\mathbf{e}_r - \mathbf{e}_s) + (\omega_{rt}-\omega_{st})\Vert\Tilde{\mathbf{e}}_t\Vert^2\right],
\end{aligned}
\end{equation}
when $\omega_{rt}-\omega_{st}\geq 0$, the expression of \eqref{20220724eq1} holds. We visualize the average aggregation weight $\omega$ of one-order neighbor users for each item groups when training LightGCN in Figure \ref{factor_omega}. From the results, it can be observed that $\omega$ is larger as degree increases generally. Therefore, the expectation of the rating difference would be enlarged for $d_r > d_s$ after graph convolution. Finish the proof.
\end{proof}

\noindent \textbf{Conclusion 2:} The theoretical analyses show that the gap of the prediction scores between popular items and tail items \textit{w.r.t.} users enlarges with deep layers of graph convolution. It indicates that popular items would become more influential by affecting more high-order neighbor users. As a consequence, popular items are more likely be over-recommended as GCNs go deeper. It reveals how GCN-based models amplify the popularity bias, providing theoretical supports to the phenomenon shown in Figure \ref{tailratio}.

\subsection{Our Method --- DAP}

In this section, we propose our method DAP (\textbf{D}ebias the \textbf{A}mplification of \textbf{P}opularity) to alleviate the issue of the popularity bias amplification of GCN in the inference stage. 

From Lemma \ref{220703thm3}, the popularity bias amplification comes from the updating of node representation and neighborhood aggregation after graph convolution when training GCN backbone models. Taking LightGCN as an example, we can quantify the bias after graph convolution at each layer in a unified form as
\begin{equation}
\begin{aligned}
\mathbf{e}^{(l)}_v &= \sum\limits_{j\in\mathcal{N}_v}\frac{1}{\sqrt{d_v}\sqrt{d_j}}\mathbf{e}^{(l-1)}_j\\
 &=  \sum\limits_{j\in\mathcal{N}_v}\frac{1}{\sqrt{d_v}\sqrt{d_j}}(\hat{\mathbf{e}}^{(l-1)}_j+ \alpha_{H_j}^{(l-1)}\pmb{\vartheta}^{(l-1)}_{H_j}+\alpha_{L_j}^{(l-1)}\pmb{\vartheta}^{(l-1)}_{L_j})\\
&=\hat{\mathbf{e}}^{(l)}_v+ \alpha_{H_v}^{(l)}\pmb{\vartheta}^{(l)}_{H_v}+\alpha_{L_v}^{(l)}\pmb{\vartheta}^{(l)}_{L_v},
\end{aligned}
\end{equation}
where $\hat{\mathbf{e}}^{(l)}_v$ is the ideal representation of the node $v$ at the $l$-th layer, $\alpha_{H_v}^{(l)}\pmb{\vartheta}^{(l)}_{H_v}$ and $\alpha_{L_v}^{(l)}\pmb{\vartheta}^{(l)}_{L_v}$ are the bias that comes from higher-degree and lower-degree neighbors respectively. Specifically, there are two amplification effects: (1) the higher-degree neighbors have large influence and dominate the updating of the target node's representation. It is inclined to make the target node position close to the higher-degree neighbors in the representation space; (2) for the lower-degree neighbors, the representations of them are influenced by the target node, leading to biased learning; after graph convolution, such bias would be further aggregated into the target node. Those two amplification effects are need to be estimated and intervened.

In addition, in order to estimate the bias of each node, we employ clustering algorithms to group node representations into clusters. The clustering can automatically discover the highly-influential nodes as they are close to each other in the representation space. For each node, we estimate its amplification effect within its cluster and then intervene the bias in the representations. The specific debiasing process is as follows.

Given a well-trained GCN-based backbone model, we could obtain the 0-th layer representations $\textbf{E}^{(0)}\in \mathbb{R}^{(M+N)\times D}$ (D is the embedding size) of all nodes. Then, these representations are fed into the next layer of graph convolution. As we discussed earlier, the bias appears after graph convolution and nodes which in the same cluster are most possible to affect each other. Thus, $Kmeans$ is employed to group nodes in the representation space. For the $l$-th layer node representations $\textbf{E}^{(l)}$ of all nodes, $Kmeans$ automatically divides them into $P$ clusters, \textit{i.e.}, 
\begin{equation}
\label{211019kmeans}
    \{\mathcal{C}_1^{(l)},\mathcal{C}_2^{(l)}, \cdots, \mathcal{C}_P^{(l)}\} = Kmeans(\textbf{E}^{(l)}),
\end{equation}
where $P$ is a hyper-parameter of $Kmeans$. For a node $v$, we can know the cluster it belongs to. To intervene the amplified bias effect $\textbf{b}_v^{(l)}$ after $l$-th graph convolution, we have the following strategy: for the node $v$ with degree $d_v$ in the cluster $\mathcal{C}^{(l)}_p$, we can obtain a set of higher-degree nodes $\mathbb{H}_v^{(l)}=\{j\in\mathcal{C}^{(l)}_p, d_j > d_v\}$, and a set of lower-degree nodes $\mathbb{L}_v^{(l)}=\{j\in\mathcal{C}^{(l)}_p, d_j < d_v\}$. For the two parts $\mathbb{H}_v^{(l)}$ and $\mathbb{L}_v^{(l)}$ of node $v$, their pooling (e.g., mean pooling, weighted average pooling with degree) representations after normalization $\
\hat{\pmb{\vartheta}}_{H_v}^{(l)} \in \mathbb{R}^{1\times D}$ and $\hat{\pmb{\vartheta}}_{L_v}^{(l)}\in \mathbb{R}^{1\times D}$ can be computed respectively. Thereafter, its amplification bias $\textbf{b}_v^{(l)}$ after $l$-th layer graph convolution is estimated by 
\begin{equation}
\label{220207bias-estimating}
    \textbf{b}_v^{(l)} = \alpha \mathcal{M}(\textbf{e}^{(l)}_v, \hat{\pmb{\vartheta}}_{H_v}^{(l)})\hat{\pmb{\vartheta}}_{H_v}^{(l)} + \beta \mathcal{M}(\textbf{e}^{(l)}_v, \hat{\pmb{\vartheta}}_{L_v}^{(l)})\hat{\pmb{\vartheta}}_{L_v}^{(l)},
\end{equation}
where $\alpha$ and $\beta$ are hyper-parameters that intervene the effect of bias in the final representation at the $l$-th layer since the popularity bias may be not completely harmful \cite{zhao2021popularity}. The larger the values of $\alpha$ and $\beta$, the greater the bias on the node. $\mathcal{M}$ is a similarity calculation function (\textit{e.g.}, cosine similarity) for measuring how strongly the node $v$ is affected by different parts.

After the above operations, we intervene the bias effect and revise the representation of the node $v$ at $l$-th layer graph convolution, \textit{i.e.},
\begin{equation}
\label{alpha_beta}
\hat{\textbf{e}}_v^{(l)}= \textbf{e}_v^{(l)} - \textbf{b}_v^{(l)}.
\end{equation}
Then the revised representation $\hat{\textbf{e}}_v^{(l)}$ is fed into the next layer in GCNs and we get the node representations at $(l+1)$-th layer $\mathbf{E}^{(l+1)}$. 
In an iterative manner, we can obtain all the ideal representations at each layer. For all representations rectified at different layers, they are assembled in the manner as the original model does to get the final representations.



\section{Experiments}
In this section, we conduct experiments to evaluate the performance of our proposed DAP, aiming at answering the following research questions:
\begin{itemize}[leftmargin=*]
\item{\textbf{RQ1}}: Does DAP outperform other debiasing methods? 
\item{\textbf{RQ2}}: How do higher-degree part, lower-degree part and the hyper-parameters affect the recommendation performance? 
\item{\textbf{RQ3}}: Can DAP mitigate the popularity bias?
\end{itemize}

\subsection{Experiments Settings}
\begin{table}[t]
\caption{Dataset description}
\resizebox{\linewidth}{!}{
\begin{tabular}{c|c|c|c|c}
\hline
Dataset     & \#Users  & \#Items  & \#Interactions & Density \\ \hline
Gowalla     & 29,858 & 40,981 & 1,027,370    & 0.00084  \\ \hline
Yelp2018    & 31,668 & 38,048 & 1,561,406    & 0.00130  \\ \hline
Amazon-book & 52,643 & 91,599 & 2,984,108    & 0.00062  \\ \hline
\end{tabular}}
\label{table:1}
\end{table}

\begin{table*}[]
    \caption{Performance comparison between our method DAP and other counterparts on the Overall and Tail test sets. The `improve' is the relative improvement of LightGCN-DAP-o over LightGCN.}
\resizebox{\linewidth}{!}{
\begin{tabular}{c|cccc|cccc|cccc}
\hline
Dataset                 & \multicolumn{4}{c|}{Gowalla}                                                                                                                         & \multicolumn{4}{c|}{Yelp2018}                                                              & \multicolumn{4}{c}{Amazon-book}                                                            \\ \hline
\multirow{2}{*}{Models} & \multicolumn{2}{c|}{Overall}                                               & \multicolumn{2}{c|}{Tail}                                               & \multicolumn{2}{c|}{Overall}                           & \multicolumn{2}{c|}{Tail}         & \multicolumn{2}{c|}{Overall}                           & \multicolumn{2}{c}{Tail}          \\ \cline{2-13} 
                        & Recall                              & \multicolumn{1}{c|}{NDCG}            & Recall                              & NDCG                              & Recall          & \multicolumn{1}{c|}{NDCG}            & Recall          & NDCG            & Recall          & \multicolumn{1}{c|}{NDCG}            & Recall          & NDCG            \\ \hline
LightGCN                & {\ul 0.1820}                        & \multicolumn{1}{c|}{{\ul 0.1546}}    & 0.0434                              & 0.0191                            & {\ul 0.0627}    & \multicolumn{1}{c|}{{\ul 0.0516}}    & 0.0091          & 0.0046          & {\ul 0.0414}    & \multicolumn{1}{c|}{{\ul 0.0321}}    & 0.009           & 0.0051          \\
BFGCN                   & 0.1083                              & \multicolumn{1}{c|}{0.0805}          & 0.0468                              & 0.0245                            & 0.0389          & \multicolumn{1}{c|}{0.0311}          & 0.0124          & 0.0076          & 0.0276          & \multicolumn{1}{c|}{0.0211}          & 0.0097          & 0.0059          \\
LightGCN-IPSCN          & \multicolumn{1}{l}{0.1325}          & \multicolumn{1}{l|}{0.1132}          & 0.0477                              & 0.0213                            & 0.0473          & \multicolumn{1}{c|}{0.0391}          & 0.0136          & 0.0077          & 0.0285          & \multicolumn{1}{c|}{0.0221}          & 0.0118          & 0.0069          \\
LightGCN-CausE          & \multicolumn{1}{l}{0.1334}          & \multicolumn{1}{l|}{0.1137}          & 0.0485                              & 0.0225                            & 0.0492          & \multicolumn{1}{c|}{0.0405}          & 0.0141          & 0.0085          & 0.0299          & \multicolumn{1}{c|}{0.0230}          & 0.0127          & 0.0078          \\
LightGCN-DICE           & 0.1337                              & \multicolumn{1}{c|}{0.1138}          & 0.0493                              & 0.0241                            & 0.0505          & \multicolumn{1}{c|}{0.0409}          & 0.0132          & 0.0073          & 0.0348          & \multicolumn{1}{c|}{0.0264}          & 0.0121          & 0.0074          \\
LightGCN-MACR           & 0.1188                              & \multicolumn{1}{c|}{0.0928}          & 0.0478                              & 0.0219                            & 0.0343          & \multicolumn{1}{c|}{0.027}           & \textbf{0.0233} & {\ul 0.0126}    & 0.0269          & \multicolumn{1}{c|}{0.0204}          & 0.0108          & 0.0065          \\
LightGCN-Tail           & 0.1647                              & \multicolumn{1}{c|}{0.1391}          & 0.0628                              & 0.0319                            & 0.057           & \multicolumn{1}{c|}{0.0466}          & 0.0154          & 0.0095          & 0.0369          & \multicolumn{1}{c|}{0.0283}          & 0.0151          & 0.0094          \\
LightGCN-BxQuAD         & 0.1378                              & \multicolumn{1}{c|}{0.1130}          & {\ul 0.0689}                        & \textbf{0.0360}                   & 0.0545          & \multicolumn{1}{c|}{0.0431}          & 0.0209          & 0.0123          & 0.0389          & \multicolumn{1}{c|}{0.0304}          & {\ul 0.0164}    & \textbf{0.0108} \\ \hline
LightGCN-DAP-o          & \multicolumn{1}{l}{\textbf{0.1834}} & \multicolumn{1}{l|}{\textbf{0.1564}} & \multicolumn{1}{l}{0.0538}          & \multicolumn{1}{l|}{0.0245}       & \textbf{0.0634} & \multicolumn{1}{c|}{\textbf{0.0521}} & 0.0137          & 0.0073          & \textbf{0.0436} & \multicolumn{1}{c|}{\textbf{0.0339}} & 0.0134          & 0.0079          \\
LightGCN-DAP-t          & \multicolumn{1}{l}{0.1672}          & \multicolumn{1}{l|}{0.1427}          & \multicolumn{1}{l}{\textbf{0.0708}} & \multicolumn{1}{l|}{{\ul 0.0354}} & 0.0562          & \multicolumn{1}{c|}{0.0461}          & {\ul 0.0218}    & \textbf{0.0129} & 0.0414          & \multicolumn{1}{c|}{0.0328}          & \textbf{0.0166} & {\ul 0.0102}    \\ \hline
improve                 & 0.77\%                              & \multicolumn{1}{c|}{1.16\%}          & 23.96\%                             & 28.27\%                           & 1.12\%          & \multicolumn{1}{c|}{0.97\%}          & 50.55\%         & 58.70\%         & 4.83\%          & \multicolumn{1}{c|}{5.61\%}          & 48.89\%         & 54.90\%         \\ \hline
\end{tabular}}
\label{table:2}
\end{table*}


We conduct experiments on three real-world datasets, Table \ref{table:1} lists the statistics of three datasets.
In order to guarantee a fair comparison, we follow the settings of LightGCN \cite{he2020lightgcn} and randomly split the training set and test set. The test set is called Overall.
Since our DAP is expected to mitigate the popularity bias and improve the performance on the tail items, we specially split a subset of tail items from the whole test set, named Tail test set, in contrast to the Overall counterpart. In addition, we randomly split 20\% data from the training set as the validation set for tuning the hyper-parameters. Note that the same splitting strategy is applied to the validation set.


\subsubsection{Compared Methods}
To evaluate the debiasing performance on recommendation, we implement our DAP with the GCN-based recommender models LightGCN and UltraGCN to explore how our DAP improves the recommendation performance for GCNs. In addition, three methods for solving the popularity bias and two methods for improving tail node representations are compared:
\begin{itemize}[leftmargin=*]
\item \textbf{BFGCN} \cite{zhao2021bilateral}: This is a novel graph convolution filter for the user-item bipartite network to improve long-tail node representations.  
\item \textbf{UltraGCN} \cite{KelongMao2021UltraGCNUS}: This is a state-of-the-art method that achieves the best performance on the three datasets. It is an ultra-simplified formulation of GCN which skips explicit message passing and directly approximate the limit of infinite graph convolution layers.
\item \textbf{IPSCN} \cite{gruson2019offline}: This method adds max-capping and normalization on IPS \cite{schnabel2016recommendations} value to reduce the variance of IPS. IPS eliminates popularity bias by re-weighting each item according to its popularity.
\item \textbf{CausE} \cite{bonner2018causal}: It requires a large sample of biased data and a small sample of unbiased data. CausE adds a regularizer term on the discrepancy between the item vectors used to fit the biased sample and their counterpart representations that fit the unbiased sample. Because there is no unbiased data in our datasets, we adopt the sampling method in \cite{zheng2021disentangling} and obtain 20\% unbiased data from the training set.
\item \textbf{DICE} \cite{zheng2021disentangling}: DICE is a method to handle with the popularity bias problem by learning causal embeddings. It is a framework with causal-specific data to disentangle interest and popularity into two sets of embeddings. 
\item \textbf{MACR} \cite{wei2021model}: This is a state-of-the-art method to eliminate the popularity bias by counterfactual reasoning. It performs counterfactual inference to remove the effect of item popularity.
\item \textbf{BxQuAD} \cite{abdollahpouri2019managing}: BxQuAD is a typical post-hoc method for improving tail item recommendations. It suffers from the recommendation accuracy drop for controlling popular items. In this paper, we adopt the Binary-xQuAD method of the original paper and set the hyper-parameter $\lambda=0.9$.
\item \textbf{Tail} \cite{liu2021tail}: This method learns a neighborhood translation from head nodes, which can be further transferred to tail nodes to enhance their representations. It is devised for node classification and we transfer it to the field of recommendation.
\end{itemize}
We report the all-ranking performance \textit{w.r.t.} two metrics: Recall and NDCG cut at 20.

\begin{table*}[]
    \caption{Performance comparison between our method DAP and other counterparts on the Overall and Tail test sets. The `improve' is the relative improvement of UltraGCN-DAP-o over UltraGCN.}
\resizebox{\linewidth}{!}{
\begin{tabular}{c|cccc|cccc|cccc}
\hline
Dataset                 & \multicolumn{4}{c|}{Gowalla}                                                                                                                                             & \multicolumn{4}{c|}{Yelp2018}                                                                                                                                  & \multicolumn{4}{c}{Amazon-book}                                                                                                                                \\ \hline
\multirow{2}{*}{Models} & \multicolumn{2}{c|}{Overall}                                                                       & \multicolumn{2}{c|}{Tail}                                           & \multicolumn{2}{c|}{Overall}                                                             & \multicolumn{2}{c|}{Tail}                                           & \multicolumn{2}{c|}{Overall}                                                             & \multicolumn{2}{c}{Tail}                                            \\ \cline{2-13} 
                        & Recall                                & \multicolumn{1}{c|}{NDCG}                                  & Recall                           & NDCG                             & Recall                           & \multicolumn{1}{c|}{NDCG}                             & Recall                           & NDCG                             & Recall                           & \multicolumn{1}{c|}{NDCG}                             & Recall                           & NDCG                             \\ \hline
UltraGCN                & \uline{0.1862} & \multicolumn{1}{c|}{\uline{0.1579}} & 0.0447                           & 0.0213                           & \uline{0.0676}  & \multicolumn{1}{c|}{\uline{0.0554}}  & 0.0127                           & 0.0074                           & \uline{0.0682}  & \multicolumn{1}{c|}{\uline{0.0556}}  & 0.0436                           & 0.0297                           \\
UltraGCN-IPSCN          & 0.1345                                & \multicolumn{1}{c|}{0.1123}                                & 0.0451                           & 0.0208                           & 0.0401                           & \multicolumn{1}{c|}{0.0324}                           & 0.0144                           & 0.0087                           & 0.0442                           & \multicolumn{1}{c|}{0.0356}                           & 0.0458                           & 0.0317                           \\
UltraGCN-CausE          & 0.1408                                & \multicolumn{1}{c|}{0.1177}                                & 0.0449                           & 0.0209                           & 0.0411                           & \multicolumn{1}{c|}{0.0329}                           & 0.0151                           & 0.0096                           & 0.0459                           & \multicolumn{1}{c|}{0.0369}                           & 0.0463                           & 0.0320                           \\
UltraGCN-DICE           & 0.1424                                & \multicolumn{1}{c|}{0.1201}                                & 0.0512                           & 0.0247                           & 0.0516                           & \multicolumn{1}{c|}{0.0417}                           & 0.0157                           & 0.0096                           & 0.0545                           & \multicolumn{1}{c|}{0.0423}                           & 0.0491                           & 0.0343                           \\
UltraGCN-MACR           & 0.1311                                & \multicolumn{1}{c|}{0.1078}                                & 0.0517                           & 0.0252                           & 0.0387                           & \multicolumn{1}{c|}{0.0323}                           & \textbf{0.0248}  & \textbf{0.0141} & 0.0501                           & \multicolumn{1}{c|}{0.0398}                           & 0.0488                           & 0.0335                           \\
UltraGCN-Tail           & 0.1788                                & \multicolumn{1}{c|}{0.1521}                                & 0.0634                           & 0.0321                           & 0.0618                           & \multicolumn{1}{c|}{0.0501}                           & 0.0167                           & 0.0102                           & 0.0599                           & \multicolumn{1}{c|}{0.0499}                           & 0.0531                           & 0.0378                           \\
UltraGCN-BxQuAD         & 0.1482                                & \multicolumn{1}{c|}{0.1289}                                & \uline{0.0694}  & \uline{0.0361}  & 0.0591                           & \multicolumn{1}{c|}{0.0482}                           & 0.0218                           & \uline{0.0136}  & 0.0623                           & \multicolumn{1}{c|}{0.0517}                           & \textbf{0.0547} & \uline{0.0386}  \\ \hline
UltraGCN-DAP-o          & \textbf{0.1868}      & \multicolumn{1}{c|}{\textbf{0.1580}}      & 0.0551                           & 0.0271                           & \textbf{0.0678} & \multicolumn{1}{c|}{\textbf{0.0555}} & 0.0135                           & 0.0079                           & \textbf{0.0688} & \multicolumn{1}{c|}{\textbf{0.0562}} & 0.0462                           & 0.0316                           \\
UltraGCN-DAP-t          & 0.1701                                & \multicolumn{1}{c|}{0.1483}                                & \textbf{0.0714} & \textbf{0.0362} & 0.0607                           & \multicolumn{1}{c|}{0.0493}                           & \uline{0.0237} & 0.0135                           & 0.0625                           & \multicolumn{1}{c|}{0.0520}                           & \uline{0.0543}  & \textbf{0.0391} \\ \hline
improve                 & 0.32\%                                & \multicolumn{1}{c|}{0.06\%}                                & 5.59\%                           & 6.57\%                           & 0.30\%                           & \multicolumn{1}{c|}{0.18\%}                           & 6.30\%                           & 6.76\%                           & 0.88\%                           & \multicolumn{1}{c|}{1.07\%}                           & 5.96\%                           & 6.40\%                           \\ \hline
\end{tabular}}
\label{table:3}
\end{table*}

\subsubsection{Hyper-parameter Settings} 
For a fair comparison, all compared GCN-based models are implemented with 3 layers except for the UltraGCN. We optimize all models with Adam \cite{kingma2014adam} with batch size as 4096. For our method, the number of clusters $P$ is searched in $\{1, 5, 10, 20, 30, \cdots, 70\}$. Note that we keep the same $P$ in each layer when operating $Kmeans$. The hyper-parameters $\alpha$ and $\beta$ in Eq. \eqref{220207bias-estimating} are tuned in the range of [0, 2.0] with step of 0.1.


\subsection{Recommendation Performance (RQ1)}
We compare all methods on the Overall and Tail test sets in Tables \ref{table:2} and \ref{table:3}, where the hyper-parameters of DAP-t and DAP-o are tuned to the best on the Tail and Overall validation sets, respectively. The promotions reported in Tables \ref{table:2} and \ref{table:3} are calculated by comparing LightGCN-DAP-o (UltraGCN-DAP-o) with LightGCN (UltraGCN). In general, our DAP significantly boosts two GCN methods on the Tail test set. The main observations are as follows:
\begin{itemize}[leftmargin=*]
    \item In all cases, our DAP-o brings performance gain in Recall and NDCG for LightGCN on the Overall and Tail test sets while other baselines only boost LightGCN on the Tail test set. These comparison methods mainly rely on suppressing popular items in exchange for the promotion of tail items. However, our method revises the node representations by intervening the popularity bias based on theoretical analyses on GCNs. It is more applicable to GCN-based backbones.
    \item In terms of the performance on the Tail test set, our DAP-t has a significant improvement over LightGCN, the average improvements of \\ LightGCN-DAP-t over LightGCN on the three datasets are 95.71\% on Recall and 121.92\% on NDCG, respectively. In the same time, the performance on the Overall test set only has a small drop. Although some competitive baselines outperform DAP-t on some metrics, DAP demonstrates stronger comprehensive abilities on different test sets. 
    \item  In order to further verify the effectiveness of our method compared to LightGCN, we show the performance comparison at each layer in Figure \ref{original_comp}. Overall, it can be seen that DAP boosts LightGCN stably layer by layer. Particularly, on the Amazon-book, when LightGCN degrades the performance as the graph convolution goes deeper, DAP has no accuracy drop. It indicates that the effectiveness of our debiasing method for improving the representations of nodes.
    \item For the method UltraGCN, we implement our DAP on it. Because UltraGCN directly uses infinite layers of graph convolution, we only can debias on its final representations. In Table \ref{table:3}, compared with other baselines, our method shows a similar trend to Table \ref{table:2}. It validates the effectiveness of our method. It should be noted that the improvement is relatively small on the Tail test set compared to that on LightGCN. This is mainly because we can not debias at each layer of graph convolution and obtain the most ideal representations.
\end{itemize}
To conclude, DAP can effectively improve the performance of tail items for GCN backbones and outperforms baselines in general.

\begin{figure*}[!]
\centering
\subfigure
   {\includegraphics[width=0.4\linewidth]{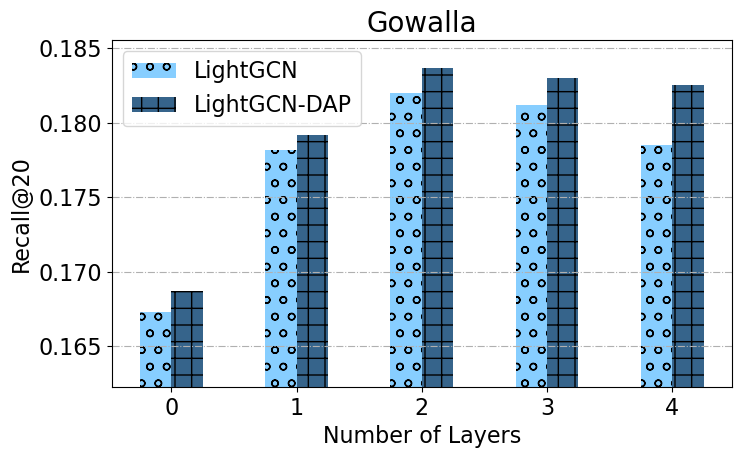}}
  {\includegraphics[width=0.4\linewidth]{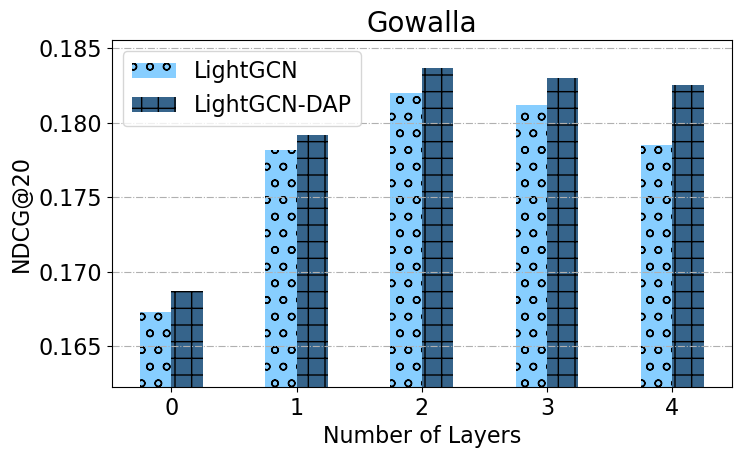}}
\subfigure
    {\includegraphics[width=0.4\linewidth]{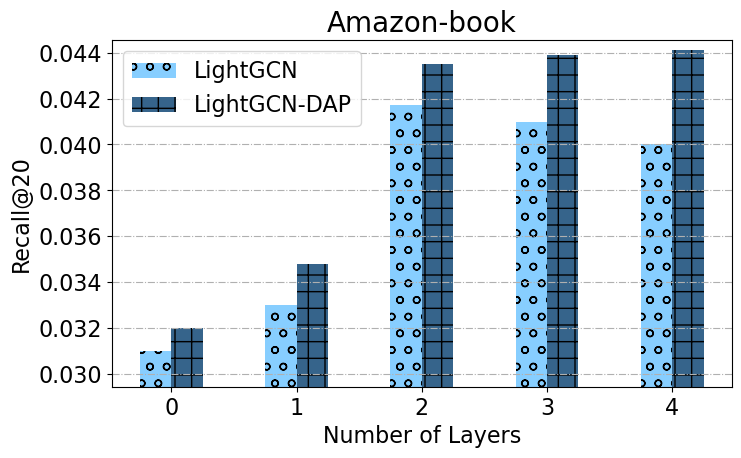}}
    {\includegraphics[width=0.4\linewidth]{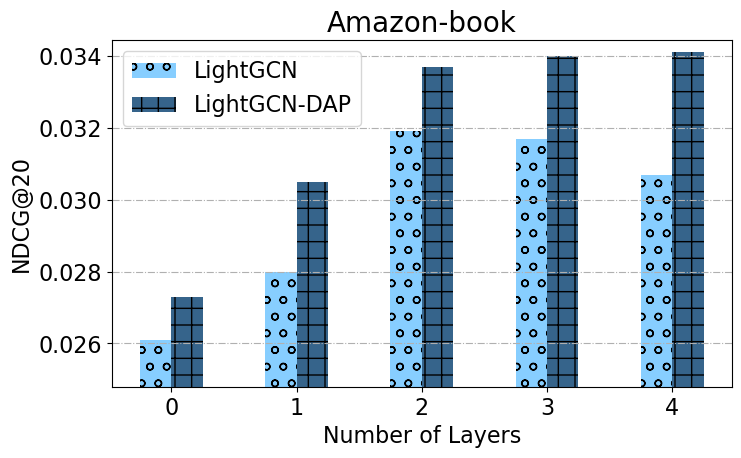}}
\caption{Performance comparison between LightGCN and LightGCN-DAP with different layers of graph convolution on the Overall test set.}
\label{original_comp}
\end{figure*}

\begin{figure*}
\centering
    {\includegraphics[width=0.7\linewidth]{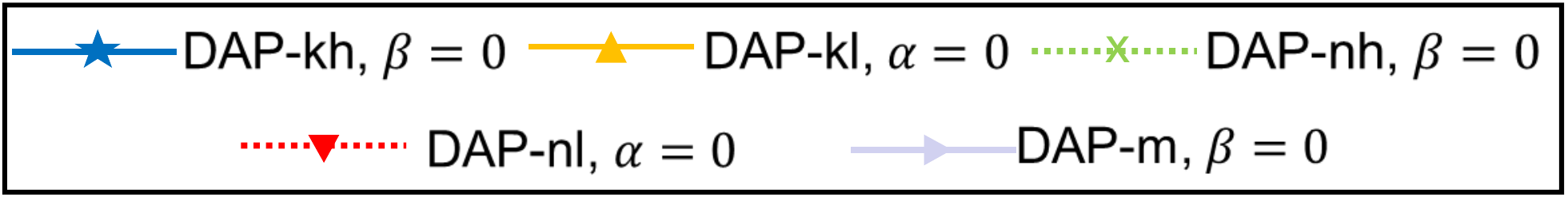}}
    {\includegraphics[width=0.3\linewidth]{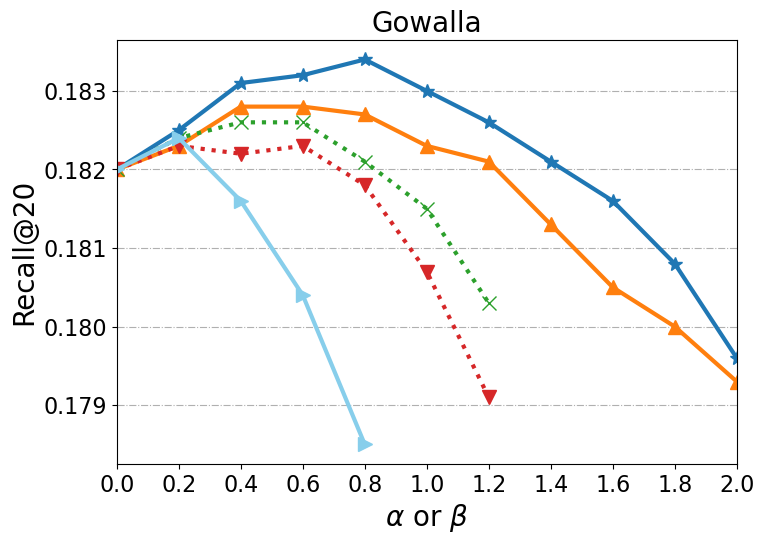}}
    {\includegraphics[width=0.3\linewidth]{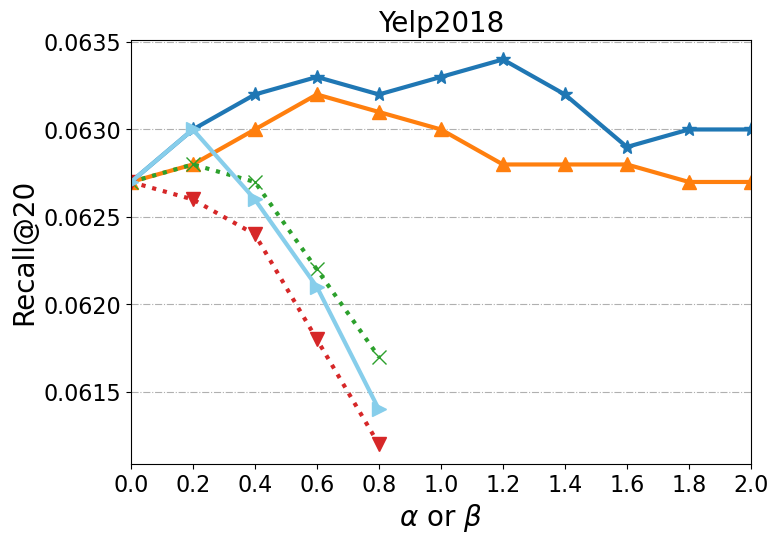}}
    {\includegraphics[width=0.3\linewidth]{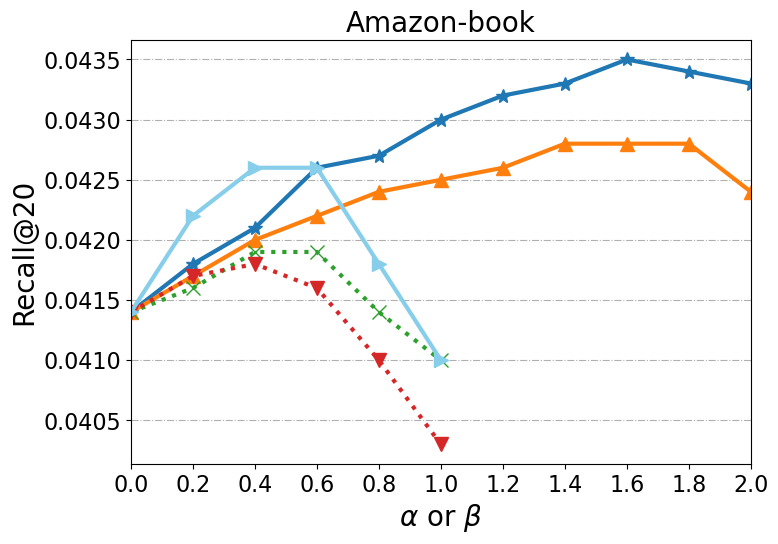}}
    {\includegraphics[width=0.3\linewidth]{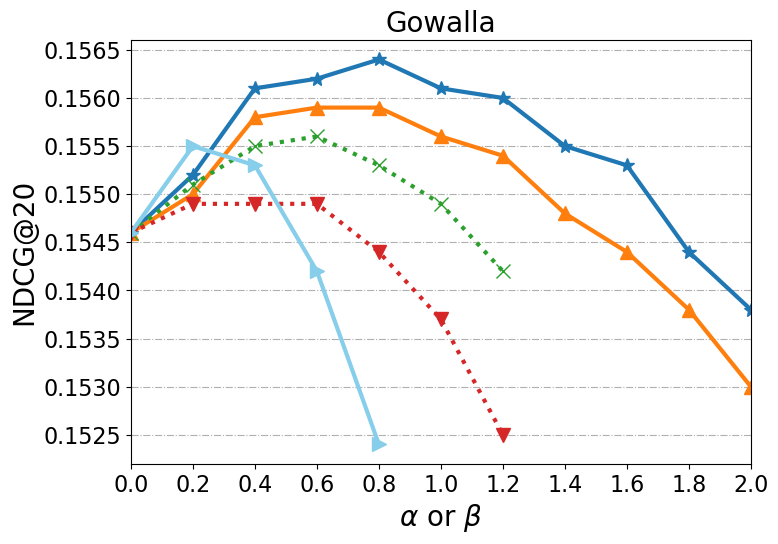}}
    {\includegraphics[width=0.3\linewidth]{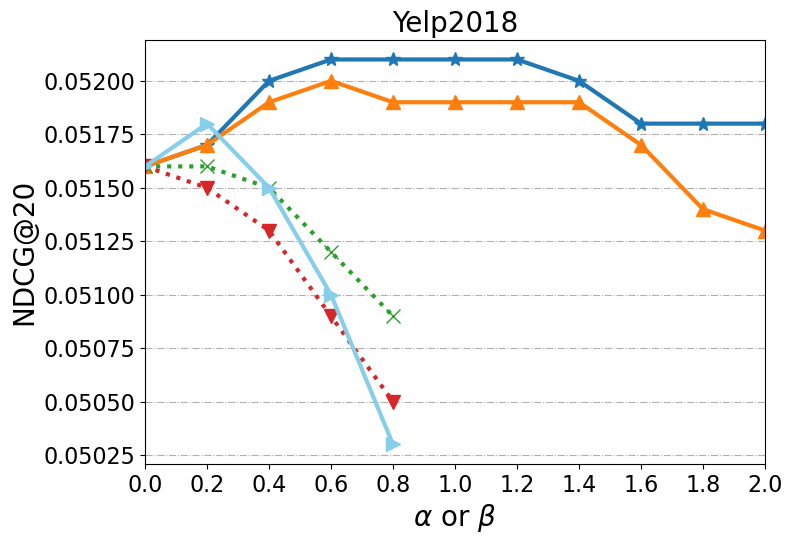}}
    {\includegraphics[width=0.3\linewidth]{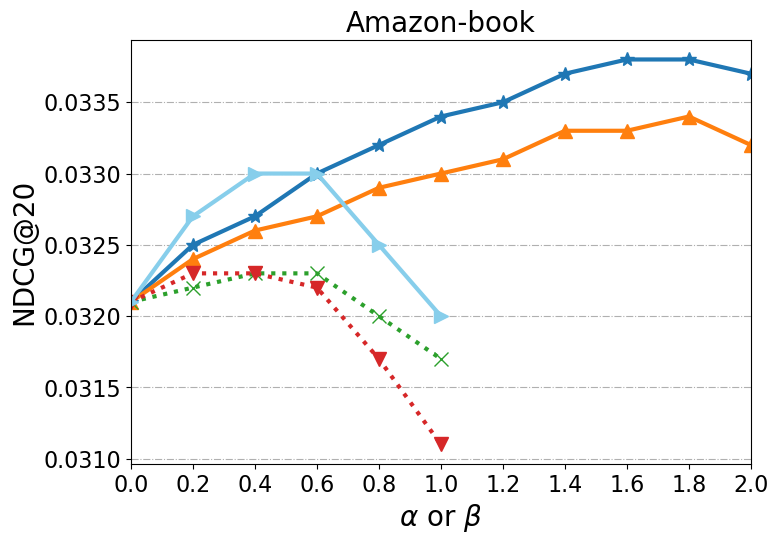}}
\caption{Ablation study of DAP with different hyper-parameters $\alpha$ and $\beta$ on the Overall test set.}  
\label{ablation}
\end{figure*}

\subsection{Ablation Study (RQ2)}

\subsubsection{Performance of Variants.}
To evaluate the amplification effect brought by higher-degree nodes and lower-degree nodes, we design five variants of DAP implemented on LightGCN including: (1) DAP-kh: this variant is derived by setting $\beta=0$ in Eq. \eqref{220207bias-estimating} which aims to evaluate the amplification effect of higher-degree nodes; (2) DAP-kl: this variant is obtained by setting $\alpha=0$ for evaluating the effect of lower-degree nodes; (3) DAP-nh: this variant treats one-order neighbors of the target node as its cluster instead of using $Kmeans$ and sets $\beta=0$ for exploring the effect of higher-degree neighbors; (4) DAP-nl: this variant is the opposite of DAP-nh, using one-order lower-degree neighbors (\textit{i.e.,} $\alpha=0$); (5) DAP-m: it is derived by removing the similarity calculation function $\mathcal{M}$ in Eq. (\ref{220207bias-estimating}) and is without considering the effect of the lower-degree part (\textit{i.e.}, $\beta=0$). 

For each variant, we need only tune $\alpha$ or $\beta$. For example, we adjust $\alpha$ for DAP-kh from 0 to 1 on the Gowalla dataset. Figure \ref{ablation} shows the results on the Overall test set and we have the following observations: (1) on the three datasets, DAP-kh achieves the best performance. It reflects that higher-degree nodes have the greater impact on other nodes than lower-degree ones. (2) DAP-kh and DAP-kl have a better performance than DAP-nh and DAP-nl on the three datasets. It indicates that estimating the amplified bias among one-order neighbors is not enough and may introduce noise. $Kmeans$ can automatically help discover the highly-influential nodes. In this way, we would estimate the amplified bias more accurately. (3) Compared to the performance of DAP-kh and DAP-m, it can be found that DAP-kh outperforms DAP-m. It reflects that the similarity function $\mathcal{M}$ could capture the relation strength among nodes and thus help estimate the amplified bias well.

\subsubsection{Effect of Different Hyper-Parameters $\alpha$ and $\beta$.}
Tuning $\alpha$ and $\beta$ is important for the performance at the inference stage. We plot the performance under different $\alpha$ and $\beta$ settings on the three datasets in Figure \ref{ablation}. It can be observed that the performance on the Overall test set increases gradually and then decreases while LightGCN is for $\alpha=0$ and $\beta=0$. This is to say, the popularity bias is not completely harmful. Directly eliminating all the bias is not reasonable in recommendation, which conforms with the finding of \cite{zhao2021popularity}.

\subsubsection{Effect of Different Hyper-Parameter $P$.}
Table \ref{table:4} reports our experimental results on the three datasets \textit{w.r.t.} the hyper-parameter $P$. As can be seen, the performance increases as $P$ becomes large and then decreases. For different datasets, when $P$ is too small or large, the popularity bias among nodes can not be captured accurately, therefore the performance is unsatisfactory. $Kmeans$ is a simple tool for clustering, other advanced unsupervised clustering tools can be explored for improving the performance.

\begin{table}[]
\caption{Effect of $P$ over DAP for the three datasets on the Overall test set.}
\resizebox{\linewidth}{!}{
\begin{tabular}{c|cc|cc|cc}
\hline
Dataset & \multicolumn{2}{c|}{Gowalla} & \multicolumn{2}{c|}{Yelp2018} & \multicolumn{2}{l}{Amazon-book} \\ \hline
P       & Recall        & NDCG         & Recall        & NDCG          & Recall         & NDCG           \\ \hline
1       & 0.1819        & 0.1547       & 0.0629        & 0.0517        & 0.0423         & 0.0329         \\ \hline
5       & 0.182         & 0.1548       & 0.063         & 0.0518        & 0.0436         & 0.0339         \\ \hline
10      & 0.1823        & 0.1552       & 0.0633        & 0.0521        & 0.0435         & 0.0339         \\ \hline
30      & 0.1831        & 0.1563       & 0.0628        & 0.0518        & 0.0435         & 0.0338         \\ \hline
50      & 0.1834        & 0.1564       & 0.0626        & 0.0517        & 0.0428         & 0.0334         \\ \hline
70      & 0.1833        & 0.1564       & 0.0623        & 0.0515        & 0.0426         & 0.0331         \\ \hline
\end{tabular}}
\label{table:4}
\end{table}

\begin{figure*}
\centering
\subfigure
    {\includegraphics[width=0.32\linewidth]{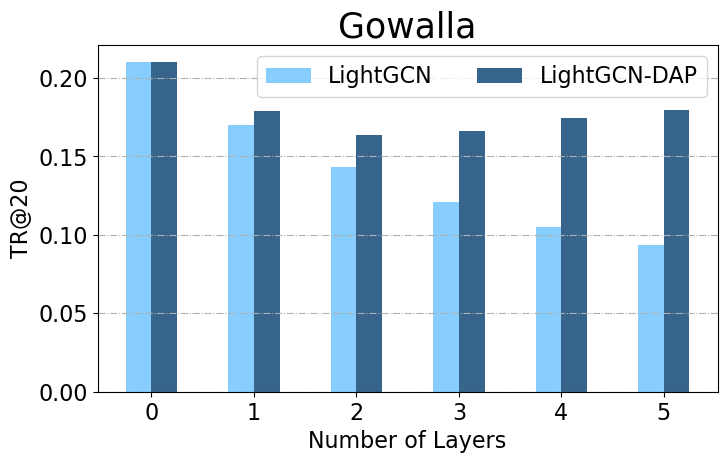}}
\subfigure
    {\includegraphics[width=0.32\linewidth]{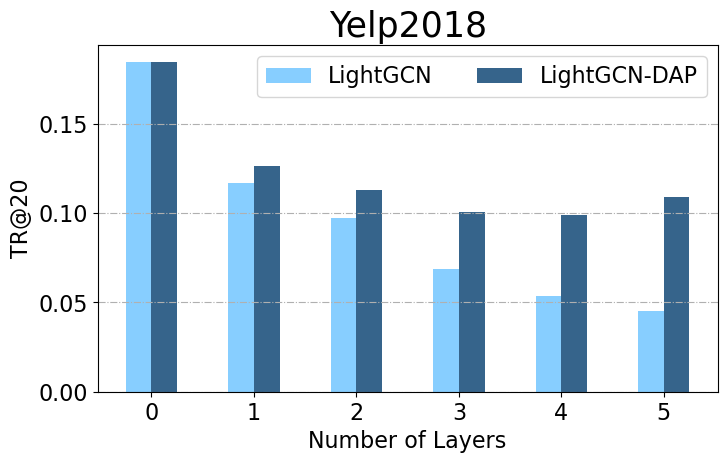}}
\subfigure
    {\includegraphics[width=0.32\linewidth]{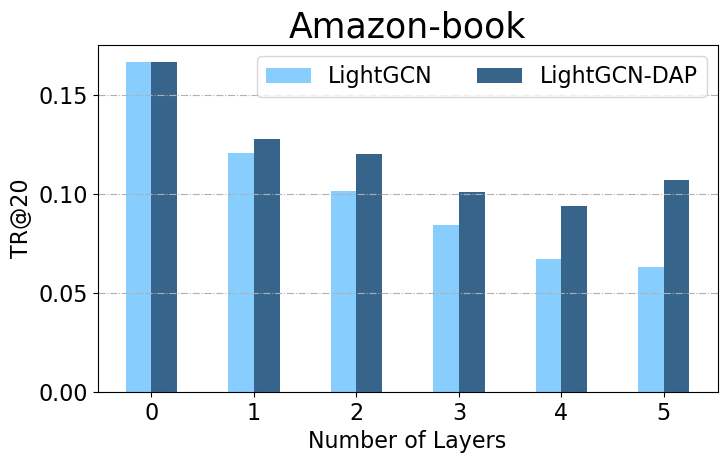}}
\caption{DAP can effectively alleviate the popularity bias.}
\label{alleviate}
\end{figure*}

\subsection{Alleviating Popularity Bias (RQ3)}
We already discuss that GCNs amplify the popularity bias in the introduction. In this part, we would verify that our debiasing framework can mitigate the amplification of popularity bias issue. 
The result of TR@20 (the ratio of recommending tail items cut at 20) is shown in Figure \ref{alleviate}. As can be seen, when the TR@20 result of LightGCN significantly decreases with more graph convolution layers, DAP could restraint continuous worsening of the popularity bias and gradually improve the TR@20 on the three datasets. It means that the tail items can be recommended more when GCNs go deeper, and thus the popularity bias is effectively alleviated by DAP. In addition, combining Figures \ref{original_comp} and \ref{alleviate} gives the conclusion that DAP can not only promote the overall performance, but also can alleviate the popularity bias at each layer.

\section{Related work}
In this section, we review the research topics related to our work: recommendation debiasing and GCNs debiasing in classification.

\subsection{Recommendation Debiasing}
Recommendation debiasing is a recently emerged research topic \cite{chen2020bias} focusing on various biases in recommendation, for example, popularity bias \cite{zhu2021popularity}\cite{xv2022neutralizing}\cite{gupta2021causer}, exposure bias \cite{gupta2021correcting}, position bias \cite{guo2019pal}, etc. Many methods have been explored to analyze and alleviate the popularity bias in recommender systems. For example, \cite{liang2016causal, schnabel2016recommendations} propose inverse propensity score (IPS) methods that reweights the interactions to debias the popularity bias in the training loss. However, these methods are difficult to tune because estimating the propensity score is challenging. \cite{dudik2011doubly} proposes another method which combines data imputation and IPS to jointly learn an unbiased recommender. However, it is a brute manner for improving the long-tail items with a huge performance drop on the recommendation accuracy. Other empirical approaches such as adversarial learning \cite{krishnan2018adversarial}, meta learning \cite{chen2021autodebias} and transfer learning \cite{zhang2021model} are developed without estimating the propensity weight. Ranking adjustment is an another line of research to solve the popularity bias issue \cite{abdollahpouri2017controlling, abdollahpouri2019managing}.   

Recently, methods based on causal inference have been widely applied to solve bias issues in recommendation, for example, MACR \cite{wei2021model} and DICE \cite{zheng2021disentangling}. And \cite{zhang2021causal} proposes a causal framework to leverage the popularity bias in recommendation. Other causal methods for learning an unbiased recommender \cite{wang2018deconfounded,bonner2018causal,sato2020unbiased,qiu2021causalrec,wang2021deconfounded} are proposed to takle various bias issues. We refer the readers to a system survey for more details \cite{chen2023bias}.

\subsection{GCNs Debiasing in Classification}
Recently, GCN-based models have demonstrated promising performance and advancements on classification tasks. However, some recent studies also reveal various issues of GCNs including over-smoothing, vulnerability and degree-related biases \cite{dai2018adversarial, xu2018powerful,li2018deeper,chen2020measuring,tang2020investigating}. The over-smoothing issue can be treated as a kind of bias in GCNs, which means the node representations are inclined to be distinguishable, and degrades the performance with many graph convolution layers. Some methods are proposed to tackle this problem. \cite{chen2020simple} proposes GCNII which extends the vanilla GCN model with initial residual connection and identity mapping to prevent the over-smoothing.
\cite{chen2020measuring} proposes two methods including adding a regularizer to the training objective and optimizing the graph topology based on the model predictions. Nevertheless, \cite{liu2020towards} argues that the over-smoothing issue only happens after a large number of iterations and the current results with several layers of graph convolution are relatively far from the ideal over-smoothing situation. For recommendation, one node can reach all the other nodes by stacking 7 layers \cite{liu2021interest}. In this regard, we could treat the nodes clustering issue with shallow graph convolution layers and over-smoothing with a large number layers differently. 

Besides, \cite{tang2020investigating} points out that GCNs are biased towards higher-degree nodes with higher accuracy than those lower-degree nodes. The authors analyze this issue and argue that nodes with low degrees tend to have very few labeled nodes, which results in sub-optimal performance on low-degree nodes. Therefore, \cite{tang2020investigating} proposes a method exploiting pseudo labels to enhance the representations for low-degree nodes. \cite{sun2020multi} also proposes a training algorithm adding confident data with virtual labels to the labeled set to enlarge the training set based on self-training. \cite{liu2021tail} proposes a method which learns a neighborhood translation from head nodes to tail nodes. In this way, the representations of tail nodes can be enhanced for improving the performance of tail nodes.  

In this work, we analyze the bias existing in GCN-based recommenders that high-degree nodes and low-degree nodes influence each other. It is different from the bias in classification introduced in \cite{tang2020investigating}. And our method is different with these methods based on producing pseudo labels and translation. 

\section{Conclusion and future work}
In this paper, we first theoretically analyzed how GCN-based recommenders amplify the popularity bias. We show that popular items tend to dominate the updating information of neighbor users in the training stage, which makes user position closer to popular items. After multiple times of neighborhood aggregation, popular items would become more influential by affecting more high-order neighbor users. Based on the above insights, we propose a simple yet generic debiasing framework DAP. Our method is implemented in each graph convolution layer in the inference stage by intervening the amplification effect on nodes in the representation space. Extensive experiments on the three real-world datasets justify the effectiveness of DAP. Our method could promote the recommendation performance on tail items and alleviate the popularity bias for GCN backbone models. 

In the future, we will explore and theoretically analyze more problems hidden in graph-based recommendation methods. 
In addition, various biases also exist in recommender systems that are harmful to users and need to be solved, such as position bias and exposure bias. 
It will be meaningful to propose a universal solution to solve various biases.

\bibliographystyle{fcs}
\bibliography{fcs}

\end{document}